\documentclass[11pt,a4paper]{amsart}
\usepackage[foot]{amsaddr}
\usepackage{fullpage}

\usepackage[T1]{fontenc}
\usepackage[utf8]{inputenc}
\usepackage[british]{babel}

\usepackage{enumerate}

\usepackage{url}
\usepackage{breakurl}
\usepackage[breaklinks]{hyperref}
\usepackage[capitalize,nameinlink,noabbrev]{cleveref}

\usepackage{xcolor}
\definecolor{lgreen}{rgb}{0.0, 0.48, 0.0}
\definecolor{lpurple}{rgb}{0.48, 0.0, 0.48}
\definecolor{bblue}{rgb}{0.2, 0.4, 0.8}
\hypersetup{linktocpage,
            colorlinks=true,
            linkcolor=lgreen,
            citecolor=lpurple,
            linktoc=true}

\definecolor{bblue}{rgb}{0.2, 0.4, 0.8}
\definecolor{bgreen}{rgb}{0.2, 0.6, 0.4}
\definecolor{bred}{rgb}{0.8, 0.4, 0.2}
\definecolor{bviolet}{rgb}{0.7, 0.2, 0.7}
\definecolor{blackred}{rgb}{0.6, 0.3, 0.3}
\definecolor{blackblue}{rgb}{0.3, 0.3, 0.6}
\usepackage{tikz}

\usepackage{eqnarray,amstext}
\usepackage{qsymbols,xparse}
\usepackage{amsmath,mathtools}
\usepackage{mathrsfs}

\theoremstyle{definition}
\newtheorem{theorem}{Theorem}[section]

\newtheorem{prop}[theorem]{Proposition}
\newtheorem{definition}[theorem]{Definition}

\newtheorem{example}[theorem]{Example}
\newtheorem*{remark}{Remark}

\usepackage{mathdots}
\usepackage[all]{xy}

\newcommand{\lf}[2]{{\raisebox{.8pc}{\tiny
\xymatrix@=.05pc{&{#2}\ar@{-}[dl]\\{#1}}}}}

\newcommand{\ri}[2]{{\raisebox{.8pc}{\tiny
\xymatrix@=.05pc{{#1}\ar@{-}[dr]\\&{#2}}}}}

\newcommand{\wh}{\textbf{when} \quad}
\newcommand{\xym}[1]{\raisebox{10pt}{$\xymatrix@C=3pt@R=5pt{#1}$}}
\newcommand{\xymm}[1]{\raisebox{18pt}{$\xymatrix@C=3pt@R=5pt{#1}$}}
\newcommand{\xyn}[1]{\raisebox{10pt}{$\xymatrix@C=2pt@R=5pt{#1}$}}

\newcommand{\idx}[1]{\mbox{\underline{\sf #1}}}
\newcommand{\deriv}[2]{\frac{\partial}{\partial #1} #2}

\newcommand{\at}[2][]{#1|_{#2}}

\let\succ\relax
\DeclareMathOperator*{\succ}{\mbox{\sf S}}

\newcommand{\cal}[1]{\mathcal{#1}}

\newcommand{\luterm}{$`l`y$\nobreakdash-term}
\newcommand{\luterms}{$`l`y$\nobreakdash-terms}
\newcommand{\lterm}{$`l$\nobreakdash-term}
\newcommand{\lterms}{$`l$\nobreakdash-terms}
\newcommand{\lcalculus}{$`l$\nobreakdash-calculus}
\newcommand{\lucalculus}{$`l`y$\nobreakdash-calculus}
\newcommand{\breduction}{$`b$\nobreakdash-reduction}

\renewcommand\autoref\Cref

\usepackage{caption, subcaption}

\begin{document}

\title{Combinatorics of explicit substitutions}
\author{Maciej Bendkowski}
\address{Theoretical Computer Science Department\\
  Faculty of Mathematics and Computer Science\\
  Jagiellonian University\\
  ul. Prof. {\L}ojasiewicza 6, 30-348 Krak\'ow, Poland.}
\email{maciej.bendkowski@tcs.uj.edu.pl}
\author{Pierre Lescanne}
\address{University of Lyon\\
        \'Ecole normale sup\'erieure de Lyon\\
        LIP (UMR 5668 CNRS ENS Lyon UCBL)\\
        46 all\'ee d'Italie, 69364 Lyon, France.}
\email{pierre.lescanne@ens-lyon.fr}
\thanks{Maciej Bendkowski was partially supported within the Polish National
Science Center grant 2016/21/N/ST6/01032.}

\begin{abstract}
    $`l`y$ is an extension of the $\lambda$-calculus which internalises the
    calculus of substitutions. In the current paper, we investigate the
    combinatorial properties of $`l`y$ focusing on the quantitative aspects of
    substitution resolution.  We exhibit an unexpected correspondence between
    the counting sequence for \luterms~and famous Catalan numbers. As a
    by-product, we establish effective sampling schemes for random \luterms. We
    show that typical \luterms~represent, in a strong sense, non-strict
    computations in the classic \lcalculus. Moreover, typically almost all
    substitutions are in fact suspended, i.e.~unevaluated, under closures.
    Consequently, we argue that $`l`y$ is an intrinsically non-strict calculus
    of explicit substitutions. Finally, we investigate the distribution of
    various redexes governing the substitution resolution in $`l`y$ and
    investigate the quantitative contribution of various substitution
    primitives.
\end{abstract}

\maketitle

\section{Introduction}
Substitution of terms for variables forms a central concept in various formal
calculi with qualifiers, such as predicate logic or different variants of
\lcalculus. Though substitution supports the computational character of
\breduction~in \lcalculus, it is usually specified as an external meta-level
formalism, see~\cite{barendregt1984}. Such an epitheoretic presentation of
substitution masks its execution as a single, indivisible calculation step, even
though it requires a considerable computational effort to carry out,
see~\cite{peytonJones1987}. In consequence, the number of \breduction~steps
required to normalise a \lterm~does not reflect the actual operational cost of
normalisation. In order to effectuate substitution resolution, its process
needs to be decomposed into a series of fine-grained atomic rewriting steps
included as part of the considered calculus.

An early example of such a calculus, internalising the evaluation of
substitution, is combinatory logic~\cite{CurryFeys1958}; alas, bearing the price
of loosing the intuitive, high-level structure of encoded functions, mirrored in
the polynomial blow-up of their representation, see~\cite{joy1984,joy1985}.
Focusing on retaining the basic intuitions behind substitution, various calculi
of explicit substitutions highlighting multiple implementation principles of
substitution resolution in \lcalculus~were proposed throughout the years,
cf.~\cite{deBruijn1978,abadi1991,Lescanne1994,DBLP:journals/jar/RoseBL12}. The
formalisation of substitution evaluation as a rewriting process provides a
formal platform for operational semantics of reduction in \lcalculus~using
abstract machines, such as, for instance, the Krivine
machine~\cite{curien93:categ_combin}. Moreover, with the internalisation of
substitution, reduction cost reflects more closely the true computational
complexity of executing modern functional programs.

Nonetheless, due to the numerous nuances regarding the evaluation cost of
functional programs (e.g.~assumed reduction or strictness strategies) supported
by a general tradition of considering computational complexity in the framework
of Turing machines or RAM models rather than formal calculi, the evaluation cost
in term rewriting systems, such as \lcalculus~or combinatory logic, gains
increasing attention only quite recently,
see~\cite{DBLP:journals/corr/abs-1208-0515,DBLP:conf/icfp/AvanziniLM15,DBLP:journals/corr/AccattoliL16}.
The continuing development of automated termination and complexity analysers for
both first\nobreakdash- and higher\nobreakdash-order term rewriting systems
echoes the immense practical, and hence also theoretical, demand for complexity
analysis frameworks of declarative programming languages, see
e.g.~\cite{aprove2014,lmcs:749}. In this context, the computational analysis of
first\nobreakdash-order term rewriting systems seems to most accurately reflect
the practical evaluation cost of declarative
programs~\cite{DBLP:conf/cade/CichonL92,DBLP:journals/jfp/BonfanteCMT01}.
Consequently, the average-case performance analysis of abstract rewriting
machines executing the declared computations requires a quantitative analysis of
their internal calculi. Such investigations provide not only key insight into
the quantitative aspects of basic rewriting principles, but also allow to
optimise abstract rewriting machines so to reflect the quantitative contribution
of various rewriting primitives.

Despite their apparent practical utility, quantitative aspects of term rewriting
systems are not well studied.  In~\cite{DBLP:journals/tcs/ChoppyKS89} Choppy,
Kaplan and Soria provide a quantitative evaluation of a general class of
confluent, terminating term rewriting systems in which the term reduction cost
(i.e.~the number of rewriting steps required to reach the final normal form) is
independent of the assumed normalisation strategy. Following a similar, analytic
approach, Dershowitz and Lindenstrauss provide an average-time analysis of
inference parallelisation in logic
programming~\cite{DBLP:conf/iclp/DershowitzL89}.  More recently, Bendkowski,
Grygiel and Zaionc analyse quantitative aspects of normal-order reduction of
combinatory logic terms and estimate the asymptotic density of normalising
combinators~\cite{bengryzai2017,BENDKOWSKI_2017}. Alas, due to the intractable,
epitheoretic formalisation of substitution in untyped \lcalculus, its
quantitative rewriting aspects have, to our best knowledge, not yet been
investigated.

In the following paper we offer a combinatorial perspective on substitution
resolution in \lcalculus~and propose a combinatorial analysis of explicit
substitutions in \lucalculus~\cite{Lescanne1994}. The paper is structured as
follows. In~\autoref{sec:preliminaries} we draft the basic characteristics of
$`l`y$~required for the reminder of the current paper. Next, we introduce the
necessary analytic toolbox used in the quantitative analysis,
see~\autoref{sec:analytic:toolbox}. In~\autoref{sec:counting:luterms} we
enumerate \luterms~and exhibit the declared correspondence between their
corresponding counting sequence and Catalan numbers. Some statistical properties
of random \luterms~are investigated in~\autoref{sec:statistical:properties}.
In~\autoref{subsec:strict:substitution:forms} we relate the typical form of
\luterms~with the classic, epitheoretic substitution tactic of \lcalculus. The
quantitative impact of substitution suspension is investigated
in~\autoref{subsec:suspended:substitutions}. In~\autoref{subsec:redexes} we
discuss the contribution of various substitution resolution primitives.  The
final~\autoref{sec:conclusions} concludes the paper.

\section{Preliminaries}\label{sec:preliminaries}

\subsection{$`l`y$-calculus}
In the current subsection we outline the main characteristics of
\lucalculus{} (lambda-upsilon calculus) required for the reminder of the paper. We refer the curious reader
to~\cite{Lescanne1994,benaissa_briaud_lescanne_rouyer-degli_1996} for a more
detailed exposition.

\begin{remark}
    We choose to outline \lucalculus~following the presentation
    of~\cite{lescanne96} where indices start with $\idx{0}$ instead
    of~\cite{Lescanne1994,benaissa_briaud_lescanne_rouyer-degli_1996} where
    de~Bruijn indices start with $\idx{1}$, as introduced by de~Bruijn himself,
    cf.~\cite{deBruijn1972}. Although both conventions are assumed in the context
    of static, quantitative aspects of \lcalculus, the former convention seems
    to be the most recent standard,
    cf.~\cite{gryles2013,gryles2015,BendkowskiGLZ16,GittenbergerGolebiewskiG16}.
\end{remark}

The computational mechanism of $`b$\nobreakdash-reduction is usually defined as
${(`l x. a) b \to_{`b} a[x := b]}$ where the right-hand side ${a[x := b]}$
denotes the epitheoretic, capture-avoiding substitution of term $b$ for variable
$x$ in $a$. \lucalculus~\cite{Lescanne1994} is a simple, first\nobreakdash-order
rewriting system internalising substitution resolution of classic \lcalculus~in
de~Bruijn notation~\cite{deBruijn1972}. Its expressions, called
\emph{($`l`y$)\nobreakdash-terms}, consist of indices $\idx{0},\idx{1},\ldots$
(for convenience also referred to as \emph{variables}), abstractions and term
application. Terms are also equipped with a new \emph{closure} operator $[s]$
denoting the ongoing resolution of substitution $s$. Explicit substitutions are
fragmented into atomic primitives of \emph{shift}, denoted as $\uparrow$, a
unary operator \emph{lift}, written as $\Uparrow$, mapping substitutions onto
substitutions, and finally a unary \emph{slash} operator, denoted as $/$,
mapping terms onto substitutions. Terms containing closures are called
\emph{impure} whereas terms without them are said to be \emph{pure}. De Bruijn
indices are encoded using a unary base expansion. In other words, $\idx{n}$ is
represented as an $n$\nobreakdash-fold application of the successor operator
$\succ$ to zero.~\autoref{fig:lambda:ups:specification} summarises the
specification of \luterms.

The rewriting rules of \lucalculus~consist of the usual $`b$\nobreakdash-rule,
specified in this framework as ${(`l a) b \to a[b/]}$ and an additional set of
(seven) rules governing the resolution of explicit substitutions,
see~\autoref{fig:lambda:ups:rewriting:system}. Remarkably, these few rewriting
rules are sufficient to correctly model $`b$\nobreakdash-reduction and also
preserve strong normalisation of closed \lterms,
see~\cite{benaissa_briaud_lescanne_rouyer-degli_1996}.  The simple syntax and
rewriting rules of \lucalculus~are not only of theoretical importance, but also
of practical interest, used as the foundation of various reduction engines.  Let us
mention that \lucalculus~and its abstract U\nobreakdash-machine executing
(strong) $`b$\nobreakdash-normalisation was successfully used as the main
reduction engine in Pollack's implementation of LEGO, a proof checker of the
Calculus of Constructions, the Edinburgh Logical Framework, and also for the
Extended Calculus of Constructions~\cite{randypollackLEGO}.

\begin{example}
    Consider the term $S = `lx.`ly.`lz.x z (y z)$. Note that in the de~Bruijn
    notation, $S$ is written as $`l `l `l \idx{2} \idx{0} (\idx{1} \idx{0})$.
    Likewise, the term $K = `lx. `ly. x$ is denoted as $`l `l \idx{1}$.
    Certainly, $K a b \to_{`b}^{+} a$ for each term $a$. Note however, that with explicit
    substitution resolution in $`l`y$, this reduction is fragmented into several
    reduction steps, as follows:
    \begin{align}\label{eq:example:reduction}
        \begin{split}
            (`l`l \idx{1}) a &\to (`l \idx{1})[a/]\\
            &\to `l (\idx{1}[\Uparrow(a/)]) \\
            &\to `l \left(\idx{0}[a/][\uparrow]\right)\\
            &\to `l \left(a[\uparrow]\right).
        \end{split}
    \end{align}
    Note that in the final reduction step of~\eqref{eq:example:reduction} we
    obtain $a[\uparrow]$. The additional shift operator guarantees that
    (potential) free indices are aptly incremented so to avoid potential variable
    captures. If $a$ is closed, i.e.~each variable in $a$ is bound, then
    $a[\uparrow]$ resolves simply to $a$, as intended.
\end{example}

\begin{figure}[hbt!]
\centering
\begin{subfigure}{.45\textwidth}
    \begin{align}
    (`l a) b &\to a [b/] \tag{Beta}\label{red:beta}\\
    (a b)[s] &\to a [s] (b [s]) \tag{App}\label{red:app}\\
    (`l a)[s] &\to `l(a[\Uparrow (s)]) \tag{Lambda}\label{red:lambda}\\
    \idx{0} [a/] &\to a \tag{FVar}\label{red:fvar}\\
    (\succ \idx{n}) [a/] &\to \idx{n} \tag{RVar}\label{red:rvar}\\
    \idx{0} [\Uparrow(s)] &\to \idx{0} \tag{FVarLift}\label{red:fvarlift}\\
    (\succ \idx{n}) [\Uparrow(s)] &\to \idx{n}[s][\uparrow] \tag{RVarLift}\label{red:rvarlift}\\
    \idx{n}[\uparrow] &\to \succ \idx{n} \tag{VarShift}\label{red:varshift}.
\end{align}
    \caption{Rewriting rules.}
    \label{fig:lambda:ups:rewriting:system}
\end{subfigure}
\begin{subfigure}{.45\textwidth}
    \begin{align}\label{eq:luterms:spec}
    \begin{split}
    \cal{T} &::= \cal{N} ~|~ `l \cal{T} ~|~ \cal{T} \cal{T} ~|~ \cal{T} [\cal{S}]\\
    \cal{S} &::= \cal{T}/ ~|~ \Uparrow(\cal{S}) ~|~ \uparrow\\
    \cal{N} &::= \idx{0} ~|~ \succ \cal{N}.
    \end{split}
\end{align}
    \caption{Terms of \lucalculus. Note that de~Bruijn indices are encoded in unary base
    using a successor operator $\succ$.}
\label{fig:lambda:ups:specification}
\end{subfigure}%
\caption{The $`l`y$-calculus rewriting system.}
\end{figure}

\section{Analytic toolbox}\label{sec:analytic:toolbox}
We base our quantitative analysis of \luterms~on techniques borrowed from
analytic combinatorics, in particular singularity analysis developed by Flajolet
and Odlyzko~\cite{FlajoletOdlyzko1990}. We refer the unfamiliar reader
to~\cite{Wilf2006,flajolet09} for a thorough introduction to (multivariate)
generating functions and analytic combinatorics.

\begin{remark}
Our arguments follow standard applications of singularity analysis to
    (multivariate) systems of generating functions corresponding to algebraic
    structures.  For the reader's convenience we offer a high-level, though
    limited to the subject of our interest, outline of this process in the
    following section.
\end{remark}

\subsection{Singularity analysis}
Interested in the quantitative properties of \luterms, for instance their
asymptotic enumeration or parameter analysis, we typically take the following
general approach.  We start the analysis with establishing a formal, unambiguous
context-free specification describing the structures of our interest. Next,
using symbolic methods~\cite[Part A.  Symbolic Methods]{flajolet09} we convert
the specification into a system of generating functions, i.e.~formal power
series $F(z) = \sum_{n \geq 0} a_n z^n$ in which the coefficient $a_n$ standing
by $z^n$, written as $[z^n]F(z)$, denotes the number of structures (objects) of
size $n$.  Interested in parameter analysis, so obtained generating functions
become bivariate and take the form $F(z,u) = \sum_{n,k \geq 0} a_{n,k} z^n u^k$
where $a_{n,k}$, also written as $[z^n u^k]F(z,u)$, stands for the number of
structures of size $n$ for which the investigated parameter takes value $k$; for
instance, $a_{n,k}$ denotes the number of terms of size $n$ with exactly
$k$ occurrences of a specific redex pattern. In this context, variable $z$
corresponds to the size of specified structures whereas $u$ is said to
\emph{mark} the investigated parameter quantities.

When the obtained system of generating functions admits an analytic solution
(i.e.~obtained formal power series are also analytic at the complex plane
origin) we can investigate the quantitative properties of respective coefficient
sequences, and so also enumerated combinatorial structures, by examining the
analytic properties of associated generating functions. The location of their
singularities, in particular so-called dominant singularities dictates the
main, exponential growth rate factor of the investigated coefficient sequence.

\begin{theorem}[Exponential growth formula~{\cite[Theorem
    IV.7]{flajolet09}}]\label{th:exponential-growth-formula}
  If $A(z)$ is analytic at the origin and $R$ is the modulus of a singularity
    nearest to the origin in the sense that
    \begin{equation}\label{eq:exp-growth-radius-of-convergence} R = \sup \{ r
        \geq 0 ~:~ A(z) \text{ is analytic in } |z| < r \}\, ,
\end{equation}
  then the coefficient $a_n = [z^n] A(z)$ satisfies
    \begin{equation}\label{eq:exp-growth-coeff} a_n = R^{-n} \theta(n) \quad
    \text{with} \quad \limsup |\theta(n)|^{\frac{1}{n}} = 1\, .\end{equation}
\end{theorem}

Generating functions considered in the current paper are algebraic, i.e.~are
branches of polynomial equations in form of $P(z,F(z)) = 0$. Since $\sqrt{z}$
cannot be unambiguously defined as an analytic functions near the origin, the
main source of singularities encountered during our analysis are roots of
radicand expressions involved in the closed-form, analytic formulae defining
studied generating functions. The following classic result due to Pringsheim
facilities the inspection of such singularities.

\begin{theorem}[Pringsheim~{\cite[Theorem~IV.6]{flajolet09}}]\label{th:pringsheim}
  If $A(z)$ is representable at the origin by a series expansion that has
    non-negative coefficients and radius of convergence $R$, then the point $z =
    R$ is a singularity of $A(z)$.
\end{theorem}

A detailed singularity analysis of algebraic generating functions, involving an
examination of the type of dominant singularities follows as a consequence of
the Puiseux series expansion for algebraic generating functions.

\begin{theorem}[Newton, Puiseux~{\cite[Theorem~VII.7]{flajolet09}}]\label{th:newton-puiseux}
  Let $F(z)$ be a branch of an algebraic function $P(z, F(z)) = 0$. Then in a
    circular neighbourhood of a singularity $\rho$ slit along a ray emanating
    from $\rho$, $F(z)$ admits a fractional Newton-Puiseux series expansion that
    is locally convergent and of the form
  \begin{equation}
      F(z) = \sum_{k \geq k_0} c_k {\left( z - \rho \right)}^{k/\kappa}\, ,
  \end{equation}
  where $k_0 \in \mathbb{Z}$ and $\kappa \geq 1$.
\end{theorem}

With available Puiseux series, the complete asymptotic expansion of sub-exponential
growth rate factors associated with coefficient sequences of investigated
algebraic generating functions can be accessed using the following standard
function scale.

\begin{theorem}[Standard function scale~{\cite[Theorem~VI.1]{flajolet09}}]\label{th:standard-func-scale}
  Let $\alpha \in \mathbb{C} \setminus \mathbb{Z}_{\leq 0}$. Then, $f(z) = {(1 -
    z)}^{-\alpha}$ admits for large $n$ a complete asymptotic expansion in form
    of
  \begin{equation} [z^n]f(z) = \frac{n^{\alpha-1}}{\Gamma(\alpha)} \left( 1 +
      \frac{\alpha(\alpha-1)}{2n} +
      \frac{\alpha(\alpha-1)(\alpha-2)(3\alpha-1)}{24n^2} + O\left(\frac{1}{n^3}\right) \right)
  \end{equation}
  where $\Gamma \colon \mathbb{C} \setminus \mathbb{Z}_{\leq 0} \to \mathbb{C}$
    is the Euler Gamma function defined as
  \begin{equation}
  \Gamma(z) = \int_{0}^{\infty} x^{z-1} e^{-x} dx.
  \end{equation}
\end{theorem}

\subsection{Parameter analysis}
Consider a random variable $X_n$ denoting a certain parameter quantity of a
(uniformly) random \luterm~of size $n$. In order to analyse the limit behaviour
of $X_n$ as $n$ tends to infinity, we utilise the moment techniques of
multivariate generating functions~\cite[Chapter 3]{flajolet09}. In particular,
if $F(z,u)$ is a bivariate generating function associated with $X_n$ where $u$
marks the considered parameter quantities, then the expectation
$\mathbb{E}(X_n)$ takes the form
\begin{equation}
    \mathbb{E}(X_n) = \dfrac{[z^n]\deriv{u}{F(z,u)}\at{u=1}}{[z^n] F(z,1)}.
\end{equation}

Consequently, the limit mean and, similarly, all higher moments can be accessed using
techniques of singularity analysis. Although such a direct approach allows to
investigate all the limit moments of $X_n$ (in particular its mean and variance) it
is usually more convenient to study the associated probability generating
function $p_n(u)$ instead, defined as
\begin{equation}
    p_n(u) = \sum_{k \geq 0} \mathbb{P}(X_n = k) u^k =
    \dfrac{[z^n]F(z,u)}{[z^n]F(z,1)}.
\end{equation}

With $p_n(u)$ at hand, it is possible to readily access the limit distribution
of $X_n$.  In the current paper we focus primarily on continuous, Gaussian limit
distributions associated with various redexes in \lucalculus. The following
Quasi-powers theorem due to Hwang~\cite{hwang1998convergence} provides means to
obtain a limit Gaussian distribution and establishes the rate at which
intermediate distributions converge to the final limit distribution.

\begin{theorem}[Quasi-powers theorem, see~{\cite[Theorem IX.8]{flajolet09}}]
\label{th:quasi:powers:theorem}
Let \( (X_n)_{n=1}^\infty \) be a sequence of non-negative discrete random
    variables (supported by $\mathbb{Z}_{\geq 0}$) with probability generating
    functions \( p_n(u) \). Assume that, uniformly in a fixed complex
    neighbourhood of $u = 1$, for sequences $\beta_n, \kappa_n \to \infty$,
    there holds
 \begin{equation}
 p_n(u) = A(u) \cdot {B(u)}^{\beta_n} \left(1 + O\left(\dfrac{1}{\kappa_n}\right) \right)
 \end{equation}
 where $A(u)$ and $B(u)$ are analytic at $u=1$ and $A(1) = B(1) = 1$.

 Assume finally that $B(u)$ satisfies the following \emph{variability
    condition}:
    \begin{equation}\label{eq:quasi:powers:variability}
 B''(1) + B'(1) - {B'(1)}^2 \neq 0.
 \end{equation}
 Then, the distribution $X_n$ is, after standardisation, asymptotically Gaussian
    with speed of convergence of order
    $O\left(\dfrac{1}{\kappa_n} + \dfrac{1}{\sqrt{\beta_n}}\right)$:
 \begin{equation}
 \mathbb{P}\left(\dfrac{X_n - \mathbb{E}(X_n)}{\sqrt{\mathbb{V}(X_n)}} \leq
     x\right) = \Phi(x) + O\left(\dfrac{1}{\kappa_n} +
     \dfrac{1}{\sqrt{\beta_n}}\right)
 \end{equation}
 where $\Phi(x)$ is the standard normal distribution function
 \begin{equation}
 \Phi(x) = \dfrac{1}{\sqrt{2 \pi}} \int_{-\infty}^x e^{-\omega^2/2} d\omega\, .
 \end{equation}
The limit expectation and variance satisfy
\begin{align}\label{eq:quasi:powers:moments}
    \begin{split}
        \mathbb E(X_n) &\sim B'(1) n\\
        \mathbb V(X_n) &\sim \left(B''(1) + B'(1) - {B'(1)}^2\right) n
    \end{split}
\end{align}
\end{theorem}

\section{Counting $`l`y$-terms}\label{sec:counting:luterms}
In current section we begin the enumeration of \luterms.  For that purpose, we
impose on them a \emph{size notion} such that the size of a \luterm, denoted as
$|\cdot|$,~is equal to the total number of constructors (in the associated term
algebra, see~\autoref{fig:lambda:ups:specification}) of which it is
built.~\autoref{fig:size:notion} provides the recursive definition of term size.

\begin{figure}[hbt!]
    \centering
\noindent\begin{minipage}{.3\linewidth}
\begin{align*}
    |\idx{n}| &= n + 1\\
    |`l a| &= 1 + |a|\\
    |a b| &= 1 + |a| + |b|\\
    |a[s]| &= 1 + |a| + |s|
\end{align*}
\end{minipage}%
\begin{minipage}{.3\linewidth}
\begin{align*}
    |a/| &= 1 + |a|\\
    |\Uparrow(s)| &= 1 + |s|\\
    |\uparrow| &= 1.
\end{align*}
\end{minipage}%
    \caption{Natural size notion for \luterms.}
    \label{fig:size:notion}
\end{figure}

\begin{remark}
Such a size notion, in which all building constructors contribute equal weight
    one to the overall term size was introduced in~\cite{Bendkowski2016} as the
    so-called \emph{natural size notion}. Likewise, we also refer to the size notion
    assumed in the current paper as natural.

    Certainly, our choice is arbitrary and, in principle, different size
    measures can be assumed,
    cf.~\cite{gryles2015,Bendkowski2016,GittenbergerGolebiewskiG16}.  For
    convenience, we choose the natural size notion thus avoiding the obfuscating
    (though still manageable) technical difficulties arising in the analysis of
    general size model frameworks, see e.g.~\cite{GittenbergerGolebiewskiG16}.
    Moreover, our particular choice exhibits unexpected consequences and hence
    is, arguably, interesting on its own, see~\autoref{prop:luterms:catalan}.
\end{remark}

Equipped with a size notion ensuring that for each $n \geq 0$ the total number
of \luterms~of size $n$ is finite, we can proceed with our enumerative analysis.
Surprisingly, the counting sequence corresponding to \luterms~in the natural
size notion corresponds also to the celebrated sequence of Catalan
numbers\protect\footnote{see~\url{https://oeis.org/A000108}.}.

\begin{prop}
    \protect\label{prop:luterms:catalan}
    Let $T(z)$ and $S(z)$ denote the generating functions corresponding to
    \luterms\\ and substitutions, respectively. Then,
    \begin{equation}\label{eq:plain:terms:genfun:final}
        T(z) = \dfrac{1-\sqrt{1-4z}}{2z} - 1 \quad \text{whereas} \quad
        S(z) = \frac{1-\sqrt{1-4z}}{2z} \left(\frac{z}{1-z}\right).
    \end{equation}
    In consequence
    \begin{equation}\label{eq:plain:terms:coeffs:final}
        [z^n]T(z) = \begin{cases}
            0, & \text{for } n = 0\\
            \dfrac{1}{n+1} \displaystyle\binom{2n}{n}, & \text{otherwise}
        \end{cases} \quad \text{and} \quad
        [z^n]S(z) = \begin{cases}
            0, & \text{for } n = 0\\
            \displaystyle\sum_{k=0}^{n-1} \dfrac{1}{k+1} \displaystyle\binom{2k}{k} &
            \text{otherwise}.
        \end{cases}
    \end{equation}
    hence also
    \begin{equation}
        [z^n]T(z) \sim \dfrac{4^n}{\sqrt{\pi} n^{3/2}} \quad \text{whereas}
        \quad [z^n]S(z) \sim \dfrac{4^{n+1}}{3 \sqrt{\pi} n^{3/2}}.
    \end{equation}
\end{prop}

\begin{proof}
    Consider the formal specification~\eqref{eq:luterms:spec} for \luterms.
    Let $N(z)$ be the generating function corresponding to de~Bruijn indices.
    Note that following symbolic methods, the generating functions $T(z)$, $S(z)$, and
    $N(z)$ give rise to the system
\begin{align}\label{eq:plain:terms:genfun:system}
    \begin{split}
        T(z) &= N(z) + z T(z) + z {T(z)}^2 + z T(z) S(z)\\
        S(z) &= z T(z) + z S(z) + z\\
        N(z) &= z + z N(z).
    \end{split}
\end{align}
Note that $N(z)$ is an independent variable
    in~\eqref{eq:plain:terms:genfun:system}.  We can therefore solve the
    equation $N(z) = z + z N(z)$ and find that ${N(z) = \dfrac{z}{1-z}}$.
    Substituting this expression for $N(z)$ in the equations defining $T(z)$ and
    $S(z)$ we obtain
    \begin{equation}\label{eq:plain:terms:genfun:system:ii}
        T(z) = \frac{z}{1-z} + z T(z) + z {T(z)}^2 + z T(z) S(z)
        \quad \text{whereas} \quad S(z) = z T(z) + z S(z) + z.
    \end{equation}
System~\eqref{eq:plain:terms:genfun:system:ii} admits two solutions, i.e.
    \begin{equation}\label{eq:plain:terms:genfun:system:iii}
        T(z) = \frac{1 \pm \sqrt{1-4 z} - 2z}{2 z} \quad \text{and} \quad
        S(z) = \frac{1 \pm \sqrt{1-4z}}{2 (1-z)},
    \end{equation}
both with agreeing signs.

    In order to determine the correct pair of generating functions we invoke the
    fact that, by their construction, both $[z^n]T(z)$ and $[z^n]S(z)$ are
    non-negative integers for all $n \geq 0$.  Consequently, the declared
    pair~\eqref{eq:plain:terms:genfun:final} is the analytic solution
    of~\eqref{eq:plain:terms:genfun:system:iii}.  At this point, we notice that
    both the generating functions in~\eqref{eq:plain:terms:genfun:final}
    resemble the famous generating function $C(z) = \dfrac{1-\sqrt{1-4z}}{2z}$
    corresponding to Catalan numbers, see e.g.~\cite[Section 2.3]{Wilf2006}.
    Indeed
    \begin{equation}\label{eq:plain:terms:genfun:ii}
        T(z) = \dfrac{1-\sqrt{1-4z}}{2z} - 1 \quad \text{whereas} \quad
        S(z) = \frac{1-\sqrt{1-4z}}{2z} \left(\frac{z}{1-z}\right).
    \end{equation}

    In this form, we can readily relate  Catalan numbers with respective
    coefficients of $T(z)$ and $S(z)$, see~\eqref{eq:plain:terms:coeffs:final}.
    From~\eqref{eq:plain:terms:genfun:ii} we obtain $T(z) = C(z) - 1$.  The
    number $[z^n]T(z)$ corresponds thus to $[z^n]C(z)$ for all $n \geq 1$ with
    the initial $[z^0]T(z) = 0$. Furthermore, given ${S(z) = C(z)
    \dfrac{z}{1-z}}$ we note that $[z^n]S(z)$ corresponds to the partial sum of
    Catalan numbers\footnote{see~\url{https://oeis.org/A014137}.} up to $n$
    (exclusively).
\end{proof}

The correspondence exhibited in~\autoref{prop:luterms:catalan} witnesses the
existence of a bijection between \luterms~of size $n$ and, for instance, plane
binary trees with $n$ inner nodes. In what follows we provide an alternative,
constructive proof of this fact.

\subsection{Bijection between $`l`y$-terms and plane binary trees}
Let $\cal{B}$ denote the set of plane binary trees (i.e.~binary trees in which
we distinguish the order of subtrees). Consider the map $\varphi \colon \cal{B}
\to \cal{T}$ defined as in~\autoref{fig:bijection}.
Note that, for convenience, we omit drawing leafs. Consequently, nodes
in~\autoref{fig:bijection} with a single or no subtrees are to be understood as
having implicit leafs attached to vacuous branches.

\begin{figure}[hbt!]
    \centering
\noindent\begin{minipage}{.45\linewidth}
\begin{displaymath}
\begin{array}{rcl@{\quad}l}
  `v\left(\xym{`(!)\ar@{-}[dr]\\ &R}\right) &=& `l`v(R)\\\\
 `v\left(\xym{&`(!)\ar@{-}[dl]\\ L}\right) &=&
 \succ\idx{n} \\\\ &&\wh `v(L) = \idx{n}\\\\
 `v\left(\xym{&`(!)\ar@{-}[dl]\\ L}\right) &=&
    a[\Uparrow^{n+1}(s)] \\\\ &&\wh `v(L) = a[\Uparrow^{n}(s)] \\\\
\end{array}
\end{displaymath}
\end{minipage}%
\begin{minipage}{.45\linewidth}
\begin{displaymath}
\begin{array}{rcl@{\quad}l}
`v(`(!)) &=& \idx{0}\\\\
`v\left(\xymm{&`(!)\ar@{-}[dl]\\ `(!)\ar@{-}[dr]\\ &R}\right)
&=& `v(R)[\uparrow] \\\\
`v\left(\xymm{&&`(!)\ar@{-}[dl]\\ &`(!)\ar@{-}[dl]\ar@{-}[dr]\\L&&R}\right)
&=& `v(L)[`v(R)/]\\\\
    `v\left(\xym{&`(!)\ar@{-}[dl]\ar@{-}[dr]\\ L&&R}\right) &=& `v(L)\,`v(R)
\end{array}
\end{displaymath}
\end{minipage}%
    \caption{Pictorial representation of the size-preserving bijection $\varphi$ between
    \luterms~and plane binary trees.}
    \label{fig:bijection}
\end{figure}

Given a tree $T$ as input, $\varphi$ translates it to a corresponding
\luterm~$\varphi(T)$ based on the shape of $T$ (performing a so-called pattern matching).
This shape, however, might be determined through a recursive call to $\varphi$,
see the second on third rule of the left-hand size of~\autoref{fig:bijection}.

\begin{prop}\label{prop:luterms:bijection}
    The map $\varphi \colon \cal{B} \to \cal{T}$ is a bijection preserving the
    size of translated structures.  In other words, given a tree $T$ with $n$
    inner nodes $\varphi(T)$ is a \luterm~of size $n$.
\end{prop}

\begin{proof}
The fact that $\varphi$ is size-preserving follows as a direct examination of
    the rules defining $\varphi$, see~\autoref{fig:bijection}. A straightforward
    structural induction certifies that all translation rules keep the size of
    both sides equal.

    In order to prove that $\varphi$ is one-to-one and onto, we note that each maximal
    (in the sense that it cannot be further continued) sequence of successive left
    branches has to terminate in either a single leaf, hence corresponding to a
    de~Bruijn index through $\varphi$, a single right turn, see the second top
    rule of the right-hand side of~\autoref{fig:bijection}, or a branching point
    consisting of a left and right turn, see the third rule of the right-hand
    side of~\autoref{fig:bijection}. A final structural induction
    finishes the proof.
\end{proof}

With a computable map $\varphi \colon \cal{B} \to \cal{T}$ it is now possible to
translate plane binary trees to corresponding \luterms~in linear, in the size of
the binary tree, time. Composing $\varphi$ with effective samplers
(i.e.~computable functions constructing random, conditioned on size, structures)
for the former, we readily obtain effective samplers for random \luterms.

We offer a Coq proof of~\autoref{prop:luterms:bijection} together with a certified
Haskell implementation of $\varphi$ in an external
repository\footnote{see
\url{https://github.com/maciej-bendkowski/combinatorics-of-explicit-substitutions}.}
with supplementary materials to the current paper.
\begin{remark}
    Using R{\'e}my's elegant sampling algorithm~\cite{Remy85} constructing
    uniformly random, conditioned on size, plane binary trees of given size $n$
    with $\varphi$ provides a linear time, exact-size sampler for \luterms.  For
    a detailed presentation of R{\'e}my's algorithm, we refer the curious reader
    to~\cite{Knuth2006,BacherBodiniJacquot2013}.  Additional combinatorial
    parameters, such as for instance the number of specific redex sub-patterns in
    sampled terms, can be controlled using the tuning techniques
    of~\cite{BendkowskiBodiniDovgal2018} developed within the general framework
    of Boltzmann samplers~\cite{Duchon2004} and the exact-size sampling
    framework of the so-called recursive method~\cite{NijenhuisWilf1978}.
\end{remark}

\section{Statistical properties of random
$`l`y$-terms}\label{sec:statistical:properties}
In the current section we focus on quantitative properties of random terms.  We
start our quest with properties of explicit substitutions within \lucalculus. In
what follows, we investigate the proportion of \luterms~representing
intermediate steps of substitution in classic \lcalculus.

\subsection{Strict substitution forms}\label{subsec:strict:substitution:forms}
When a $`b$\nobreakdash-rule is applied and $(`lx. a) b$ is rewritten to $a[x :=
b]$ the meta-level substitution of $b$ for variable $x$ in $a$ is executed
somewhat outside of the calculus. In operational terms, the substitution $a[x := b]$ is
meant to be resolved ceaselessly and cannot be, for instance, suspended or even
(partially) omitted if it produces a dispensable result. Such a
resolution tactic is reflected in \lucalculus~in terms of the following notion
of strict substitution forms.

\begin{definition}
    A \luterm~$t$ is in \emph{strict substitution form} if there exist two pure
    (i.e.~without explicit substitutions) terms $a,b$ and a sequence
    $t_1,\ldots,t_n$ of \luterms~such that
    \begin{equation}
    a[b/] \to t_1 \to \cdots \to t_n = t
    \end{equation}
    and none of the above reductions is (Beta).

    Otherwise, $t$ is said to be in \emph{lazy substitution form}.
\end{definition}

In other words, strict substitution forms represent the intermediate computations
of resolving substitutions in the classic \lcalculus. Certainly, by design
\lucalculus~permits more involved resolution tactics, mixing for instance
$Beta$\nobreakdash-reduction and $`y$\nobreakdash-reductions. In the following
proposition we show that the proportion of terms representing the indivisible,
classic resolution tactic tends to zero with the term size tending to infinity.

\begin{prop}
Asymptotically almost all $`l\upsilon$-terms are in lazy substitution form.
\end{prop}

\begin{proof}
    We argue that the set of strict substitution forms is asymptotically
    negligible in the set of all \luterms. Consequently, its compliment,
    i.e.~lazy substitution forms, admits an asymptotic density one, as claimed.

    Consider the class $\overline{\cal{T}}$ of \luterms~containing nested
    substitutions, i.e.~subterms in form of $a[b/]$ where $b$ is impure.  Note
    that if $t$ contains nested substitutions, then $a$ cannot be in strict
    substitution form. Let us therefore estimate the asymptotic density of terms
    without nested substitutions.  Following the combinatorial
    specification~\eqref{eq:luterms:spec} for \luterms~we can write down the
    following specification for $\overline{\cal{T}}$ using the auxiliary classes
    $\overline{\cal{S}}$ of (restricted) substitutions, and $\cal{P}$ of pure
    \luterms:
    \begin{align}\label{eq:lazy:sub:forms:spec}
        \begin{split}
            \overline{\cal{T}} &::= \cal{N} ~|~ `l \overline{\cal{T}} ~|~
        \overline{\cal{T}} \overline{\cal{T}} ~|~ \overline{\cal{T}} [\overline{\cal{S}}]\\
            \overline{\cal{S}} &::= \cal{P}/ ~|~ \Uparrow(\overline{\cal{S}}) ~|~ \uparrow\\
            \cal{P} &::= \cal{N} ~|~ `l \cal{P} ~|~ \cal{P} \cal{P}.
        \end{split}
    \end{align}
    Note that~\eqref{eq:lazy:sub:forms:spec} is almost identical
    to~\eqref{eq:luterms:spec} except for the fact that we permit only pure
    terms under the slash operator in the definition of $\overline{\cal{S}}$. We
    can now apply symbolic methods and establish a corresponding system of
    generating functions:
    \begin{align}\label{eq:lazy:sub:forms:system}
        \begin{split}
            \overline{T}(z) &= N(z) + z \overline{T}(z) + z {\overline{T}(z)}^2
            + z \overline{T}(z) \overline{S}(z)\\
            \overline{S}(z) &= z P(z) + z \overline{S}(z) + z\\
            P(z) &= N(z) + z P(z) + z {P(z)}^2.
        \end{split}
    \end{align}
    Solving~\eqref{eq:lazy:sub:forms:system} for $\overline{T}(z)$ we find that
    \begin{equation}\label{eq:lazy:sub:genfun}
        \overline{T}(z) = \frac{1 - z - z \overline{S}(z) - \sqrt{\left(1 -z -z
        \overline{S}(z)\right)^2-\frac{4 z^2}{1-z}}}{2 z}
    \end{equation}
    whereas
    \begin{equation}\label{eq:lazy:sub:genfun:ii}
        \overline{S}(z) = \frac{z + z P(z)}{1 -z} \quad \text{and} \quad
        P(z) = \frac{1-z-\sqrt{\left(1-z\right)^2-\frac{4z^2}{1-z}}}{2z}.
    \end{equation}

    Note that both $\overline{S}(z)$ and $P(z)$ share a common, unique dominant
    singularity $\rho \doteq 0.295598$ being the smallest positive root of the radicand
    expression $\left(1-z\right)^2-\dfrac{4z^2}{1-z}$ in the defining formula of
    $P(z)$, see~\eqref{eq:lazy:sub:genfun:ii}. Moreover, due to the presence of
    the expression $z \overline{S}(z)$ in the numerator
    of~\eqref{eq:lazy:sub:genfun} $\rho$ is also a singularity of $\overline{T}(z)$.
    Denote the radicand expression of~\eqref{eq:lazy:sub:genfun} as $R(z)$.
    Note that
    \begin{equation}
        \dfrac{d}{d z} R(z) =  - 2 \left(1 - z - z \overline{S}(z)\right)
        \left(1 + \overline{S}(z) + z\dfrac{d}{d z} \overline{S}(z)\right)-\frac{4 z^2}{(1-z)^2}-\frac{8
   z}{1-z}.
    \end{equation}
    Since $\overline{S}(z)$ is a generating function with non-negative integer
    coefficients both $\overline{S}(z)$ and its derivative
    $\dfrac{d}{d z} \overline{S}(z)$ are positive in the interval
    $z \in \left(0,\rho\right)$. Therefore, the derivative
    $\dfrac{d}{d z} R(z)$ is negative for suitable values $z$
    satisfying $1 - z - z \overline{S}(z) \geq 0$. A direct
    computation verifies that $0 < z \leq \rho$ satisfy this condition.
    Consequently, $R(z)$ is decreasing in the interval $z \in \left(0,\rho\right)$.

    At this point we note that $R(0) = 1$ whereas $R(\frac{1}{4}) > 0$. It
    follows therefore that $R(z)$ has no roots in the interval
    $(0,\frac{1}{4})$. Since the generating function $T(z)$ corresponding to
    all (unrestricted) \luterms~has a single dominant singularity $\zeta =
    \frac{1}{4}$, see~\eqref{eq:plain:terms:genfun:final}, a straightforward
    application of the exponential growth formula
    (see~\autoref{th:exponential-growth-formula}) reveals that \luterms~without
    nested substitutions are asymptotically negligible in the set of all
    \luterms. So is, as well, its subset of strict substitution forms.
\end{proof}

\subsection{Suspended substitutions}\label{subsec:suspended:substitutions}
Closures in \luterms~are intended to represent suspended, unevaluated
substitutions in the classic \lcalculus. In other words, substitutions whose
resolution is meant to be carried out in a non-strict manner. It the current
section we investigate the quantitative impact of this suspension on random
terms.

\begin{definition}
    Let $s$ be an substitution and $t$ be a \luterm. Then, $s$, all its
    subterms and, all the constructors it contains
are said to be \emph{suspended in $t$} if $t$ contains a subterm in
    form of $[s]$; in other words, when $s$ occurs under a closure in $t$.
\end{definition}

In the following proposition we show that, in expectation, almost all of the
term content (i.e.~represented computation) is suspended under closures.

\begin{prop}
Let $X_n$ be a random variable denoting the number of constructors not suspended
    under a closure in a random \luterm~of size $n$. Then, the expectation
    $\mathbb{E}(X_n)$ satisfies
    \begin{equation}\label{eq:closures:mean:final}
        \mathbb{E}(X_n) \xrightarrow[n\to\infty]{} \dfrac{316}{3}.
    \end{equation}
\end{prop}

\begin{proof}
    Let $T(z,u)$ be a bivariate generating function where $[z^n u^k]T(z,u)$
    denotes the number of \luterms~of size $n$ with $k$ constructors not
    suspended under a closure. In other words, the number of \luterms~of size
    $n$ in which suspended substitutions are of total size $n - k$.

    Given the specification~\eqref{eq:luterms:spec} for $\cal{T}$ we introduce a
    new variable $u$ and mark with it constructors which do not occur under
    closures. Note that, in doing so, we no not mark $\cal{T}$ recursively in
    $\cal{S}$.  Following symbolic methods we note that $T(z,u)$ satisfies
    \begin{equation}\label{eq:closures:spec}
        T(z,u) = \frac{z u}{1- z u} + z u T(z,u) + z u {T(z,u)}^2 + z u T(z,u)
        S(z).
    \end{equation}
    Taking the derivative $\partial u$ at $u = 1$ of both sides
    of~\eqref{eq:closures:spec} we arrive at
    \begin{align}\label{eq:closures:deriv}
        \begin{split}
            \deriv{u}{T(z,u)}\at{u=1} &= \dfrac{z}{\left(1- z\right)^2} + T(z,1) \bigg(z + z T(z,1) + z S(z)
            \bigg) \\& \qquad\qquad+ \deriv{u}{T(z,u)}\at{u=1} \bigg(z + 2 z T(z,1) + z S(z) \bigg)
        \end{split}
    \end{align}
    and hence
    \begin{equation}\label{eq:closures:deriv:ii}
        \deriv{u}{T(z,u)}\at{u=1} = \dfrac{\dfrac{z}{\left(1- z\right)^2} + T(z,1) \bigg(z + z T(z,1) + z S(z)
        \bigg)}{1 - z - 2 z T(z,1) - z S(z)}.
    \end{equation}

    From~\eqref{eq:plain:terms:genfun:final} we note that both $T(z)$ and $S(z)$
    admit Puiseux expansions in form of $`a - `b\sqrt{1-4z}+ O(\big|1-4z\big|)$
    for some appropriate (different for $T(z)$
    and $S(z)$) constants $`a,`b > 0$. Consequently, both
    the numerator and denominator of~\eqref{eq:closures:deriv} admit
    Puiseux expansions of similar form. Furthermore, as
    \begin{equation}
        \dfrac{a - b \sqrt{1-4z} + O\left(\big|1-4z\big|\right)}{c - d
        \sqrt{1-4z}+ O\left(\big|1-4z\big|\right)} =
         \left(\frac{a d}{c^2-4 d^2 z+d^2}-\frac{b c}{c^2-4 d^2
        z+d^2}\right) \sqrt{1-4 z} + O\left(\big|1-4z\big|\right)
    \end{equation}
    we conclude that
    \begin{equation}
        \deriv{u}{T(z,u)}\at{u=1} = `g - `d \sqrt{1-4z} +
        O\left(\big|1-4z\big|\right)
    \end{equation}
    near $z = \frac{1}{4}$ for some (computable) constants $`g,`d > 0$.

    An application of the standard function scale
    (see~\autoref{th:standard-func-scale}) provides now the asymptotic
    estimate
    \begin{equation}
        \mathbb{E}(X_n) = \dfrac{[z^n]\deriv{u}{T(z,u)}\at{u=1}}{[z^n]T(z,1)}
        \xrightarrow[n\to\infty]{} C.
    \end{equation}
    A direct calculation gives the specific
    quantity~\eqref{eq:closures:mean:final} of $C$.
\end{proof}

\subsection{Substitution resolution primitives}\label{subsec:redexes}
The internalisation of substitution in \lucalculus~introduces several new types
of redexes governing the resolution of closures,
see~\autoref{fig:lambda:ups:rewriting:system}. Instead of a single
$`b$\nobreakdash-redex, specific implementations of the \lucalculus~rewriting
system, such as for instance the abstract U-machine, have to handle eight
rewriting rules together with their intricate interaction.

In the current section we investigate the distribution of specific redexes in
random \luterms, providing insight in the quantitative contribution of various
substitution resolution primitives. Since all redexes share virtually the same
proof scheme, for convenience, we provide detailed arguments only for the (Beta)
rule.  Remaining proofs are merely sketched.

\subsubsection{(Beta) redexes}
\begin{prop}
Let $X_n$ be a random variable denoting the number of $`b$\nobreakdash-redexes
    in a random \luterm~of size $n$. Then, after standardisation, $X_n$
    converges in law to a Gaussian distribution with speed of convergence of
    order $O\left(\frac{1}{\sqrt{n}}\right)$. The limit expectation
    $\mathbb{E}(X_n)$ and variance $\mathbb{V}(X_n)$ satisfy
\begin{equation}\label{eq:redex:beta:mean:variance:final}
    \mathbb{E}(X_n) \xrightarrow[n\to\infty]{} \frac{3}{64} n \quad \text{and}
    \quad \mathbb{V}(X_n) \xrightarrow[n\to\infty]{} \frac{153}{4096} n
\end{equation}
\end{prop}

\begin{proof}
Routinely, we begin our considerations with
establishing a combinatorial specification for corresponding
$`b$\nobreakdash-redexes. Note that given the specification $\cal{T}$ for
general \luterms~\eqref{eq:luterms:spec} we can rewrite the left operand in the
$(\cal{T} \cal{T})$ production using the recursive specification for $\cal{T}$.
Consequently, we obtain the following modified combinatorial specification:
\begin{align}\label{eq:redex:beta:spec}
\begin{split}
\cal{T} &::= \cal{N}~|~`l \cal{T}~|~\cal{T} [\cal{S}]~|~\overbrace{\cal{N}
    \cal{T}~|~\underline{(`l \cal{T}) \cal{T}}~|~(\cal{T} [\cal{S}])
    \cal{T}~|~(\cal{T} \cal{T})\cal{T}}^{(\cal{T}\cal{T})}\\
    \cal{S} &::= \cal{T}/ ~|~ \Uparrow(\cal{S}) ~|~ \uparrow\\
    \cal{N} &::= \idx{0} ~|~ \succ \cal{N}.
\end{split}
\end{align}
Note that in this form, specification~\eqref{eq:redex:beta:spec} explicitly uses
the production $(`l \cal{T}) \cal{T}$ associated with $`b$\nobreakdash-redexes. Since the above
specification is unambiguous, i.e.~each \luterm~has precisely one derivation
starting from $\cal{T}$, we can further convert it into the following system of
corresponding bivariate generating functions marking $`b$\nobreakdash-redexes:
\begin{align}\label{eq:redex:beta:gfun}
\begin{split}
    T(z,u) &= \dfrac{z}{1-z} + z T(z,u) + z T(z,u) S(z,u) + {z T(z,u)}^2 + (u-1)z^2
    {T(z,u)}^2\\
    S(z,u) &= z T(z,u) + z S(z,u) + z.
\end{split}
\end{align}
Note that instead of a direct transformation, we use the, somewhat
indirect, $z^2 {T(z,u)}^2 \left(u-1\right)$ expression to denote
\luterms~in application form where each $`b$\nobreakdash-redex is
marked with variable~$u$.  It should be understood as this: we count
first the whole $(\cal{T}\cal{T})$ by $ z T(z,u)^2$, then we mark the
$`b$-redexes by $u z^2 T(z,u)^2$, then we remove the over-counted
$`b$-redexes by $- z^2 T(z,u)^2$.  At this point, we
solve~\eqref{eq:redex:beta:gfun} and find that the generating function
$T(z,u)$ satisfies
\begin{equation}\label{eq:redex:beta:gfun:ii}
    T(z,u) = \frac{1-z-\frac{z^2}{1-z} -
    \sqrt{\left(1-z-\frac{z^2}{1-z}\right)^2 -\frac{4 z^2 \left(1+(u-1)
    z+\frac{z}{1-z}\right)}{1-z} }}{2z
    \left(1+ (u-1) z+\frac{z}{1-z}\right)}.
\end{equation}

In this form it is clear that the dominant singularity $\rho(u)$ is carried by
the radicand expression of $T(z,u)$, see~\eqref{eq:redex:beta:gfun:ii}.
Furthermore, the singularity $\rho(u)$ is non-constant and \emph{moving},
i.e.~varies smoothly with $u$. Its specific form can be readily accessed by
equating the above radicand expression to zero and noting that
\begin{equation}\label{eq:redex:beta:rho}
    \rho(u) = \frac{1}{4} + \frac{1}{2} \sqrt{\frac{\sqrt[3]{u+3}}{2^{2/3}
   (u-1)^{2/3}}+\frac{1}{4}}-\frac{1}{2} \sqrt{\frac{u+7}{4 (u-1)
   \sqrt{\frac{\sqrt[3]{u+3}}{2^{2/3}
   (u-1)^{2/3}}+\frac{1}{4}}}-\frac{\sqrt[3]{u+3}}{2^{2/3}
   (u-1)^{2/3}}+\frac{1}{2}}
\end{equation}
is the only solution satisfying $\lim_{u\to1} \rho(u) =
\frac{1}{4}$ corresponding to the dominant singularity of $T(z,1)$.
    Consequently, $\rho(u)$ can be analytically continued
    onto a larger domain containing the point $u = 1$. Let $\rho(u)$ denote, by
    a slight abuse of notation, this continuation of~\eqref{eq:redex:beta:rho}.
It follows that $T(z,u)$ can be uniquely represented as
\begin{equation}
    T(z,u) = `a (z,u) + `b (z,u) \sqrt{1-\dfrac{z}{\rho(u)}}
\end{equation}
where both $`a(z,u)$ and $`b(z,u)$ are non-vanishing near $(z,u) =
(\frac{1}{4},1)$. With $u$ fixed sufficiently close to one, we can now apply the
standard function scale (see~\autoref{th:standard-func-scale}) and obtain the
estimate
\begin{equation}
    [z^n] T(z,u) = `g \left(\dfrac{1}{\rho(u)}\right)^n \quad \text{with} \quad
    `g = \dfrac{`b\left(\rho(u),u\right)}{2 \sqrt{\pi} n^{3/2}}.
\end{equation}
Consequently, the probability generating function $p_n(u)$ satisfies
\begin{equation}
    p_n(u) = \dfrac{[z^n] T(z,u)}{[z^n] T(z,1)} = \overline{`g}
    \left(\dfrac{\rho(1)}{\rho(u)}\right)^n \left(1 + O\left(\dfrac{1}{n}\right) \right)
    \quad \text{where} \quad \overline{`g} =
    \dfrac{`b\left(\rho(u),u\right)}{`b\left(\rho,1\right)}.
\end{equation}
Such a form of $p_n(u)$ matches the premises of
the Quasi-powers theorem (see~\autoref{th:quasi:powers:theorem}) taking
\begin{equation}
    A(u) =  \dfrac{\beta(\rho(u), u)}{\beta(\rho, 1)}\, ,
    \quad B(u) = \dfrac{\rho(1)}{\rho(u)}
    \quad  \text{and} \quad
    \beta_n = \kappa_n = n.
\end{equation}
Given the explicit formula~\eqref{eq:redex:beta:rho} for $\rho(u)$ a routine
calculation verifies the requested variability
condition~\eqref{eq:quasi:powers:variability}.  Consequently, an application of the
Quasi-powers theorem finishes the proof. The limit expectation and
variance~\eqref{eq:redex:beta:mean:variance:final} associated with $X_n$ can be
computed using formulas~\eqref{eq:quasi:powers:moments}.
\end{proof}

\subsubsection{(App) redexes}
In order to mark occurrence of (App) redexes, we take the
specification~\eqref{eq:luterms:spec} of $\cal{T}$ and rewrite the production
$\cal{T}[\cal{S}]$ into four, more detailed ones, including an explicit
production for (App) redexes.
\begin{align}\label{eq:redex:app:spec}
\begin{split}
\cal{T} &::= \cal{N}~|~`l \cal{T}~|~\cal{T} \cal{T}~|~\overbrace{\cal{N}
    [\cal{S}]~|~(`l \cal{T}) [\cal{S}]~|~(\cal{T} [\cal{S}])
    [\cal{S}]~|~\underline{(\cal{T} \cal{T})[\cal{S}]}}^{\cal{T}[\cal{S}]}\\
    \cal{S} &::= \cal{T}/ ~|~ \Uparrow(\cal{S}) ~|~ \uparrow\\
    \cal{N} &::= \idx{0} ~|~ \succ \cal{N}.
\end{split}
\end{align}
A direct translation onto the level of corresponding generation functions gives
\begin{align}\label{eq:redex:app:gfun}
\begin{split}
    T(z,u) &= \dfrac{z}{1-z} + z T(z,u) + z {T(z,u)}^2 + z T(z,u) S(z,u)
        + (u-1) z^2 {T(z,u)}^2 S(z,u)\\
    S(z,u) &= z T(z,u) + z S(z,u) + z.
\end{split}
\end{align}
A detailed analysis provides then the following result.
\begin{prop}
    Let $X_n$ be a random variable denoting the number of (App)\nobreakdash-redexes
    in a random \luterm~of size $n$. Then, after standardisation, $X_n$
    converges in law to a Gaussian distribution with speed of convergence of
    order $O\left(\frac{1}{\sqrt{n}}\right)$. The limit expectation
    $\mathbb{E}(X_n)$ and variance $\mathbb{V}(X_n)$ satisfy
\begin{equation}\label{eq:redex:app:mean:variance:final}
    \mathbb{E}(X_n) \xrightarrow[n\to\infty]{} \frac{1}{32} n \quad \text{and}
    \quad \mathbb{V}(X_n) \xrightarrow[n\to\infty]{} \frac{45}{2048} n
\end{equation}
\end{prop}
\subsubsection{(Lambda) redexes}
    The case of (Lambda) redexes is virtually identical to (App) redexes. This
    time, however, a different production is marked:
\begin{align}\label{eq:redex:lambda:spec}
\begin{split}
\cal{T} &::= \cal{N}~|~`l \cal{T}~|~\cal{T} \cal{T}~|~\overbrace{\cal{N}
    [\cal{S}]~|~\underline{(`l \cal{T}) [\cal{S}]}~|~(\cal{T} [\cal{S}])
    [\cal{S}]~|~(\cal{T} \cal{T})[\cal{S}]}^{\cal{T}[\cal{S}]}\\
    \cal{S} &::= \cal{T}/ ~|~ \Uparrow(\cal{S}) ~|~ \uparrow\\
    \cal{N} &::= \idx{0} ~|~ \succ \cal{N}.
\end{split}
\end{align}
Accordingly,
\begin{align}\label{eq:redex:lambda:gfun}
\begin{split}
    T(z,u) &= \dfrac{z}{1-z} + z T(z,u) + z {T(z,u)}^2 + z T(z,u) S(z,u)
        + (u-1) z^2 {T(z,u)} S(z,u)\\
    S(z,u) &= z T(z,u) + z S(z,u) + z.
\end{split}
\end{align}
\begin{prop}
    Let $X_n$ be a random variable denoting the number of (Lambda)\nobreakdash-redexes
    in a random \luterm~of size $n$. Then, after standardisation, $X_n$
    converges in law to a Gaussian distribution with speed of convergence of
    order $O\left(\frac{1}{\sqrt{n}}\right)$. The limit expectation
    $\mathbb{E}(X_n)$ and variance $\mathbb{V}(X_n)$ satisfy
\begin{equation}\label{eq:redex:lambda:mean:variance:final}
    \mathbb{E}(X_n) \xrightarrow[n\to\infty]{} \frac{1}{32} n \quad \text{and}
    \quad \mathbb{V}(X_n) \xrightarrow[n\to\infty]{} \frac{53}{2048} n
\end{equation}
\end{prop}
\subsubsection{(FVar) redexes}
The case of (FVar) redexes is a bit more involved and requires three layers of
production substitution in order to reach an explicit (FVar) production.
\begin{align}\label{eq:redex:fvar:spec}
\begin{split}
\cal{T} &::= \cal{N}~|~`l \cal{T}~|~\cal{T} \cal{T}~|~\overbrace{
    \overbrace{\underbrace{\underline{\idx{0} [\cal{T}/]}~|~\idx{0}
    [\Uparrow(\cal{S})]~|~\idx{0}
    [\uparrow]}_{\idx{0}[\cal{S}]}~|~(\succ \cal{N})[\cal{S}]}^{\cal{N}[\cal{S}]}~|~(`l \cal{T}) [\cal{S}]~|~(\cal{T} [\cal{S}])
    [\cal{S}]~|~(\cal{T} \cal{T})[\cal{S}]}^{\cal{T}[\cal{S}]}\\
    \cal{S} &::= \cal{T}/ ~|~ \Uparrow(\cal{S}) ~|~ \uparrow\\
    \cal{N} &::= \idx{0} ~|~ \succ \cal{N}.
\end{split}
\end{align}
The corresponding system of generating functions takes then the form
\begin{align}\label{eq:redex:fvar:gfun}
\begin{split}
    T(z,u) &= \dfrac{z}{1-z} + z T(z,u) + z {T(z,u)}^2 + z T(z,u) S(z,u)
        + (u-1) z^3 {T(z,u)}\\
    S(z,u) &= z T(z,u) + z S(z,u) + z.
\end{split}
\end{align}
\begin{prop}
    Let $X_n$ be a random variable denoting the number of (FVar)\nobreakdash-redexes
    in a random \luterm~of size $n$. Then, after standardisation, $X_n$
    converges in law to a Gaussian distribution with speed of convergence of
    order $O\left(\frac{1}{\sqrt{n}}\right)$. The limit expectation
    $\mathbb{E}(X_n)$ and variance $\mathbb{V}(X_n)$ satisfy
\begin{equation}\label{eq:redex:fvar:mean:variance:final}
    \mathbb{E}(X_n) \xrightarrow[n\to\infty]{} \frac{3}{256} n \quad \text{and}
    \quad \mathbb{V}(X_n) \xrightarrow[n\to\infty]{} \frac{729}{65536} n
\end{equation}
\end{prop}
\subsubsection{(RVar) redexes}
Similarly to (FVar) redexes, (RVar) redexes require three layers of
substitution. The final layer, however, involves now $(\succ \cal{N})
[\cal{S}]$.
\begin{align}\label{eq:redex:rvar:spec}
\begin{split}
\cal{T} &::= \cal{N}~|~`l \cal{T}~|~\cal{T} \cal{T}~|~\overbrace{
    \overbrace{\idx{0}[\cal{S}]~|~\underbrace{\underline{(\succ
    \cal{N})[\cal{T}/]}~|~(\succ
    \cal{N})[\Uparrow(\cal{S})]~|~(\succ \cal{N})[\uparrow]}_{(\succ
    \cal{N})[\cal{S}]}}^{\cal{N}[\cal{S}]}~|~(`l \cal{T}) [\cal{S}]~|~(\cal{T} [\cal{S}])
    [\cal{S}]~|~(\cal{T} \cal{T})[\cal{S}]}^{\cal{T}[\cal{S}]}\\
    \cal{S} &::= \cal{T}/ ~|~ \Uparrow(\cal{S}) ~|~ \uparrow\\
    \cal{N} &::= \idx{0} ~|~ \succ \cal{N}.
\end{split}
\end{align}
When transformed, we obtain the following system of generating functions:
\begin{align}\label{eq:redex:rvar:gfun}
\begin{split}
    T(z,u) &= \dfrac{z}{1-z} + z T(z,u) + z {T(z,u)}^2 + z T(z,u) S(z,u)
        + (u-1) \dfrac{z^4}{1-z} {T(z,u)}\\
    S(z,u) &= z T(z,u) + z S(z,u) + z.
\end{split}
\end{align}
\begin{prop}
    Let $X_n$ be a random variable denoting the number of (RVar)\nobreakdash-redexes
    in a random \luterm~of size $n$. Then, after standardisation, $X_n$
    converges in law to a Gaussian distribution with speed of convergence of
    order $O\left(\frac{1}{\sqrt{n}}\right)$. The limit expectation
    $\mathbb{E}(X_n)$ and variance $\mathbb{V}(X_n)$ satisfy
\begin{equation}\label{eq:redex:rvar:mean:variance:final}
    \mathbb{E}(X_n) \xrightarrow[n\to\infty]{} \frac{1}{256} n \quad \text{and}
    \quad \mathbb{V}(X_n) \xrightarrow[n\to\infty]{} \frac{249}{65 536} n
\end{equation}
\end{prop}
\subsubsection{(FVarLift) redexes}
The (FVarLift) redex follows the same successive transformation of the initial
specification $\cal{T}$. When (FVarLift) redexes are obtained and marked, we
get the following outcome specification:
\begin{align}\label{eq:redex:fvarlift:spec}
\begin{split}
\cal{T} &::= \cal{N}~|~`l \cal{T}~|~\cal{T} \cal{T}~|~\overbrace{
    \overbrace{\underbrace{\idx{0} [\cal{T}/]~|~\underline{\idx{0}
    [\Uparrow(\cal{S})]}~|~\idx{0}
    [\uparrow]}_{\idx{0}[\cal{S}]}~|~(\succ \cal{N})[\cal{S}]}^{\cal{N}[\cal{S}]}~|~(`l \cal{T}) [\cal{S}]~|~(\cal{T} [\cal{S}])
    [\cal{S}]~|~(\cal{T} \cal{T})[\cal{S}]}^{\cal{T}[\cal{S}]}\\
    \cal{S} &::= \cal{T}/ ~|~ \Uparrow(\cal{S}) ~|~ \uparrow\\
    \cal{N} &::= \idx{0} ~|~ \succ \cal{N}.
\end{split}
\end{align}
Consequently
\begin{align}\label{eq:redex:fvarlift:gfun}
\begin{split}
    T(z,u) &= \dfrac{z}{1-z} + z T(z,u) + z {T(z,u)}^2 + z T(z,u) S(z,u)
        + (u-1) z^3 {S(z,u)}\\
    S(z,u) &= z T(z,u) + z S(z,u) + z.
\end{split}
\end{align}
\begin{prop}
    Let $X_n$ be a random variable denoting the number of (FVarLift)\nobreakdash-redexes
    in a random \luterm~of size $n$. Then, after standardisation, $X_n$
    converges in law to a Gaussian distribution with speed of convergence of
    order $O\left(\frac{1}{\sqrt{n}}\right)$. The limit expectation
    $\mathbb{E}(X_n)$ and variance $\mathbb{V}(X_n)$ satisfy
\begin{equation}\label{eq:redex:fvarlift:mean:variance:final}
    \mathbb{E}(X_n) \xrightarrow[n\to\infty]{} \frac{1}{128} n \quad \text{and}
    \quad \mathbb{V}(X_n) \xrightarrow[n\to\infty]{} \frac{241}{32768} n
\end{equation}
\end{prop}
\subsubsection{(RVarLift) redexes}
(RVarLift) redexes are specified analogously to (FVarLift) redexes.
The deepest level of transformation involved now $(\succ \cal{N}) [\cal{S}]$
instead of $\idx{0}[\cal{S}]$ as in the case of (FVarLift) redexes.
The resulting specification takes the form
\begin{align}\label{eq:redex:rvarlift:spec}
\begin{split}
\cal{T} &::= \cal{N}~|~`l \cal{T}~|~\cal{T} \cal{T}~|~\overbrace{
    \overbrace{\idx{0}[\cal{S}]~|~\underbrace{(\succ
    \cal{N})[\cal{T}/]~|~\underline{(\succ
    \cal{N})[\Uparrow(\cal{S})]}~|~(\succ \cal{N})[\uparrow]}_{(\succ
    \cal{N})[\cal{S}]}}^{\cal{N}[\cal{S}]}~|~(`l \cal{T}) [\cal{S}]~|~(\cal{T} [\cal{S}])
    [\cal{S}]~|~(\cal{T} \cal{T})[\cal{S}]}^{\cal{T}[\cal{S}]}\\
    \cal{S} &::= \cal{T}/ ~|~ \Uparrow(\cal{S}) ~|~ \uparrow\\
    \cal{N} &::= \idx{0} ~|~ \succ \cal{N}.
\end{split}
\end{align}
And so, the associated system of generating function becomes
\begin{align}\label{eq:redex:rvarlift:gfun}
\begin{split}
    T(z,u) &= \dfrac{z}{1-z} + z T(z,u) + z {T(z,u)}^2 + z T(z,u) S(z,u)
        + (u-1) \frac{z^4}{1-z} {S(z,u)}\\
    S(z,u) &= z T(z,u) + z S(z,u) + z.
\end{split}
\end{align}
\begin{prop}
    Let $X_n$ be a random variable denoting the number of (RVarLift)\nobreakdash-redexes
    in a random \luterm~of size $n$. Then, after standardisation, $X_n$
    converges in law to a Gaussian distribution with speed of convergence of
    order $O\left(\frac{1}{\sqrt{n}}\right)$. The limit expectation
    $\mathbb{E}(X_n)$ and variance $\mathbb{V}(X_n)$ satisfy
\begin{equation}\label{eq:redex:rvarlift:mean:variance:final}
    \mathbb{E}(X_n) \xrightarrow[n\to\infty]{} \frac{1}{384} n \quad \text{and}
    \quad \mathbb{V}(X_n) \xrightarrow[n\to\infty]{} \frac{377}{147456} n
\end{equation}
\end{prop}
\subsubsection{(VarShift) redexes}
The final case of (VarShift) redexes can be approached as before. Marking
$\cal{N}[\uparrow]$ in the associated specification we obtain
\begin{align}\label{eq:redex:varshift:spec}
\begin{split}
\cal{T} &::= \cal{N}~|~`l \cal{T}~|~\cal{T} \cal{T}~|~\overbrace{
    \underbrace{
    \cal{N}[\cal{T}/]~|~
        \cal{N}[\Uparrow(\cal{S})]~|~\underline{\cal{N}[\uparrow]}}_{
    \cal{N}[\cal{S}]}~|~(`l \cal{T}) [\cal{S}]~|~(\cal{T} [\cal{S}])
    [\cal{S}]~|~(\cal{T} \cal{T})[\cal{S}]}^{\cal{T}[\cal{S}]}\\
    \cal{S} &::= \cal{T}/ ~|~ \Uparrow(\cal{S}) ~|~ \uparrow\\
    \cal{N} &::= \idx{0} ~|~ \succ \cal{N}.
\end{split}
\end{align}
And so
\begin{align}\label{eq:redex:varshift:gfun}
\begin{split}
    T(z,u) &= \dfrac{z}{1-z} + z T(z,u) + z {T(z,u)}^2 + z T(z,u) S(z,u)
        + (u-1) \frac{z^3}{1-z}\\
    S(z,u) &= z T(z,u) + z S(z,u) + z.
\end{split}
\end{align}
\begin{prop}
    Let $X_n$ be a random variable denoting the number of (VarShift)\nobreakdash-redexes
    in a random \luterm~of size $n$. Then, after standardisation, $X_n$
    converges in law to a Gaussian distribution with speed of convergence of
    order $O\left(\frac{1}{\sqrt{n}}\right)$. The limit expectation
    $\mathbb{E}(X_n)$ and variance $\mathbb{V}(X_n)$ satisfy
\begin{equation}\label{eq:redex:varshift:mean:variance:final}
    \mathbb{E}(X_n) \xrightarrow[n\to\infty]{} \frac{1}{64} n \quad \text{and}
    \quad \mathbb{V}(X_n) \xrightarrow[n\to\infty]{} \frac{57}{4096} n
\end{equation}
\end{prop}

\begin{remark}
    In order to facilitate the intensive calculations involved in obtaining
    redex distributions, one can use formulas~\eqref{eq:quasi:powers:moments}
    and the implicit form of the singularity $\rho(u)$ defined as a root of an
    appropriate radicand expression $P(z,u)$ of the corresponding bivariate
    generating function $T(z,u)$. The quantities $\rho'(1)$ and $\rho''(1)$ can
    be extracted using implicit derivatives of the equation $P(\rho(u),u) = 0$.
\end{remark}

\autoref{tab:redex:outline} outlines the obtained means and variances for all
considered $`l`y$\nobreakdash-redexes.

\begin{table}[hbt!]
\caption{List of redex distributions in descending order of limit means.
    Specific values are approximated to the sixth decimal point.}
\label{tab:redex:outline}
\begin{tabular}{|c|l|l|}
  \hline
    Redex & Mean & Variance \\
  \hline
    (Beta) & $0.046875 n$ & $0.037354 n$ \\
    (App) & $0.031250 n$ & $0.021973 n$ \\
    (Lambda) & $0.031250 n$ & $0.025879 n$ \\
    (VarShift) & $0.015625 n$ & $0.013916 n$ \\
    (FVar) & $0.011719 n$ & $0.011124 n$ \\
    (FVarLift) & $0.007812 n$ & $0.007355 n$ \\
    (RVar) & $0.003906 n$ & $0.003799 n$ \\
    (RVarLift) & $0.002604 n$ & $0.002557 n$ \\
  \hline
\end{tabular}
\end{table}

\section{Conclusions}\label{sec:conclusions}
Our contribution is a step towards the quantitative analysis of substitution
resolution and, in particular, the average-case analysis of abstract machines
associated with calculi of explicit substitutions. Although we focused on
\lucalculus, other calculi are readily amenable to similar analysis. Our
particular choice is motivated by the relative, compared to other calculi of
explicit substitutions, simple syntax of $`l`y$. With merely eight rewriting
rules, $`l`y$ is one of the conceptually simplest calculi of explicit
substitutions. Notably, rewriting rules contribute just to the technical part of
the quantitative analysis, not its general scheme. Consequently, we expect that
investigations into more complex calculi might be more technically challenging,
however should not pose significantly more involved issues.

Our quantitative analysis exhibited that typical \luterms~represent, in a strong
sense, intrinsically non-strict computations of the classic \lcalculus.
Typically, substitutions are not ceaselessly evaluated, but rather suspended in
their entirety; almost all of the encoded computation is suspended under
closures. Not unexpectedly, on average, the most frequent redex is (Beta). In
the $`y$ fragment of $`l`y$, however, the most recurrent redexes are, in order,
(App) and (Lambda). The least frequent, and at the same time the most intricate
redex, is (RVarLift). Let us note that such a diversity of redex frequencies
might be exploited in practical implementations. For instance, knowing that
specific redexes are more frequent than others, abstract machines might be
aptly optimised.

Finally, as an unexpected by-product of our analysis, we exhibited a
size-preserving bijection between \luterms~and plane binary trees, enumerated by
the famous Catalan numbers. Notably, such a correspondence has practical
implications. Specifically, we established an exact-size sampling scheme for
random \luterms~based on known samplers for the latter structures.  Consequently,
it is possible to effectively generate random \luterms~of size $n$ in $O(n)$
time.

\section*{Acknowledgements}
We would like to thank Sergey Dovgal for fruitful discussions and valuable
comments.

\bibliographystyle{alpha}
\bibliography{references}

\end{document}